\documentclass[11pt]{article}
\usepackage{epsfig}
\newenvironment{proof}{\noindent {\bf Proof.}}{\bigskip}

\usepackage{epsfig}
\usepackage{amsmath}
\usepackage{amssymb}
\usepackage{color}

\usepackage{tabularx}
\usepackage{epsfig}
\usepackage{amsmath}
\usepackage{amssymb}
\usepackage{graphicx,color, algorithmic, algorithm}
\usepackage{color}
\usepackage{verbatim}
\usepackage{boxedminipage}

{ %
  \newtheorem{lemma}{\textbf{Lemma}}[section]%
  \newtheorem{theorem}[lemma]{\textbf{Theorem}}%
}

\usepackage[cyr]{aeguill}
\usepackage{epsfig}
\usepackage{amsmath}
\usepackage{amssymb}
\usepackage{graphicx,color, algorithmic, algorithm}
\usepackage{color}

\begin{document}

\centerline{\Large\bf Minimum Manhattan network problem in normed planes}

\centerline{\Large\bf with polygonal balls: a factor 2.5 approximation
algorithm}

\vspace{5mm}

\centerline{\large {\sc N. Catusse, V. Chepoi, K. Nouioua, and Y. Vax\`es}}

\vspace{5mm}

\begin{center}
Laboratoire d'Informatique Fondamentale de Marseille,\\[0.1cm]
Facult\'{e} des Sciences de Luminy, Aix-Marseille Universit\'e,\\[0.1cm]
 F-13288 Marseille Cedex 9, France\\[0.1cm]
{\em\{{catusse,chepoi,nouioua,vaxes\}@lif.univ-mrs.fr}}
\end{center}

\vspace{5mm}

\begin{abstract} Let ${\mathcal B}$ be a centrally symmetric convex polygon of
${\mathbb R}^2$ and $||{\bf p}-{\bf q}||$ be the distance between
two points ${\bf p},{\bf q}\in {\mathbb R}^2$ in the normed plane
whose unit ball is ${\mathcal B}$. For a set $T$ of $n$ points
(terminals) in ${\mathbb R}^2$, a ${\mathcal B}$-{\it Manhattan
network} on $T$ is a network $N(T)=(V,E)$ with the property that its
edges are parallel to the directions of ${\mathcal B}$  and for
every pair of terminals ${\bf t}_i$ and ${\bf t}_j$,  the network
$N(T)$ contains a shortest ${\mathcal B}$-path between them, i.e., a
path of length $||{\bf t}_i-{\bf t}_j||.$ A {\it minimum ${\mathcal
B}$-Manhattan network} on $T$ is a ${\mathcal B}$-Manhattan network
of minimum possible length. The problem of finding minimum
${\mathcal B}$-Manhattan  networks has been introduced  by
Gudmundsson, Levcopoulos, and Narasimhan (APPROX'99) in the case
when the unit ball ${\mathcal B}$ is a square (and hence the
distance $||{\bf p}-{\bf q}||$ is the $l_1$ or the
$l_{\infty}$-distance between ${\bf p}$ and ${\bf q}$)  and it has
been shown recently by Chin, Guo, and Sun \cite{ChGuSu} to be
strongly NP-complete. Several approximation algorithms (with factors
8,4,3, and 2) for the minimum Manhattan problem are known. In this
paper, we propose a factor 2.5 approximation algorithm for the minimum
${\mathcal B}$-Manhattan network problem. The algorithm employs a
simplified version of the strip-staircase decomposition proposed in our
paper \cite{ChNouVa} and subsequently used in other factor 2
approximation algorithms for the minimum Manhattan problem.
\end{abstract}


 \section{Introduction}

\subsection{Normed planes}  Given a compact, centrally symmetric,
convex set ${\mathcal B}$ in the plane ${\mathbb R}^2,$ one can define a norm
$||\cdot||:=||\cdot||_{{\mathcal B}}: {\mathbb R}^2\rightarrow
{\mathbb R}^+$ by setting $||{\bf v}||=\lambda,$ where ${\bf
v}=\lambda {\bf u}$ and ${\bf u}$ is a unit vector belonging to the
boundary of ${\mathcal B}.$ We can then define a metric
$d:=d_{\mathcal B}$ on ${\mathbb R}^2$ by setting $d({\bf p},{\bf
q})=||{\bf p}-{\bf q}||.$ The resulting metric space $({\mathbb
R}^2,d_{\mathcal B})$ is called a {\it normed} (or {\it Minkowski})
plane with unit disk (gauge) ${\mathcal B}$ \cite{BoMaSo,Th}.
%We will denote by ${\bf p}{\bf q}$ the line segment joining two points ${\bf p}$ and ${\bf q}.$
In this paper, we consider normed planes in which the unit ball
${\mathcal B}$ is a centrally symmetric convex polygon (i.e., a {\it
zonotope}) of ${\mathbb R}^2.$ We denote by ${\bf b}_0,\ldots {\bf
b}_{2m-1}$ the $2m$ vertices of ${\mathcal B}$ (in counterclockwise
order around the circle) as well as the $2m$ unit vectors that
define these vertices. By central symmetry of ${\mathcal B},$ ${\bf
b}_k=-{\bf b}_{k+m}$ for $k=0,\ldots, m-1.$ A {\it legal
$k$-segment} of $({\mathbb R}^2,d_{\mathcal B})$ is a segment  ${\bf
p}{\bf q}$ lying on a line parallel to the line passing via ${\bf
b}_k$ and ${\bf b}_{k+m}.$ A {\it legal path} $\pi({\bf p},{\bf q})$
between two points ${\bf p},{\bf q}$ of ${\mathbb R}^2$ is any path
connecting ${\bf p}$ and ${\bf q}$ in which all edges are legal
segments.   The length of $\pi({\bf p},{\bf q})$ is the sum of
lengths of its edges. A {\it shortest ${\mathcal B}$-path} between
${\bf p}$ and ${\bf q}$ is a legal $({\bf p},{\bf q})$-path of
minimum length. The best known example of normed planes with
polygonal unit balls is the $l_1$-plane (also called the rectilinear
plane) with norm $||{\bf v}||=|x({\bf v})|+|y({\bf v})|.$ The unit
ball of the $l_1$-plane is a square whose two diagonals lie on the
$x$-axis and $y$-axis, respectively. The {\it $l_1$-distance}
between two points ${\bf p}$ and ${\bf q}$ is $d({\bf p},{\bf
q}):=||{\bf p}-{\bf q}||_1=|x({\bf p})-x({\bf q})|+|y({\bf
p})-y({\bf q})|.$ The legal paths of the rectilinear plane are the
paths consisting of horizontal and vertical segments, i.e.,
rectilinear paths. Another important particular case of polygonal
norms is that of $\lambda$-norms (alias uniform norms)
\cite{BrZa,BrThWeZa} for which the unit ball ${\mathcal B}$ is a
regular polygon.

\subsection{Minimum Manhattan and ${\mathcal B}$-Manhattan network problems}

 A {\it rectilinear network} $N=(V,E)$  in ${\mathbb R}^2$ consists of a
 finite set $V$ of points  and  horizontal and vertical
 segments connecting pairs of points of $V.$  The {\it length}  of  $N$
 is the sum of lengths of its edges.  Given a finite set $T$
 of points in the plane, a {\it Manhattan network}
 \cite{GuLeNa} on $T$ is a rectilinear network $N(T)=(V,E)$ such
 that $T\subseteq V$ and for every pair of points in $T,$ the
 network $N(T)$ contains a shortest rectilinear path between them. A
 {\it minimum Manhattan network} on $T$ is a Manhattan network of minimum
 possible length and the Minimum Manhattan Network  problem ({\it MMN problem})
 is to find such a network.

More generally, given a zonotope ${\mathcal B}\subset {\mathbb
R}^2,$ a ${\mathcal B}$-network $N=(V,E)$ consists of a finite set
$V$ of points  and  legal segments connecting pairs of points of $V$
(the edges of $N$). The {\it length} $l(N)$ of $N$ is the sum of
lengths of its edges.  Given a set $T=\{ {\bf t}_1,\ldots,{\bf
t}_n\}$ of $n$ points (called {\it terminals}),   a {\it ${\mathcal
B}$-Manhattan network} on $T$ is a ${\mathcal B}$-network
$N(T)=(V,E)$ such that $T\subseteq V$ and for every pair of
terminals in $T,$ the
 network $N(T)$ contains a shortest ${\mathcal B}$-path between them (see Fig. \ref{bman}). A
 {\it minimum ${\mathcal B}$-Manhattan network} on $T$ is a ${\mathcal
B}$-Manhattan network of minimum
 possible length and the Minimum ${\mathcal B}$-Manhattan Network  problem ({\it $\mathcal B$-MMN problem})
 is to find such a network. Fig. \ref{bman_bis} illustrates the evolution of a minimum ${\mathcal
B}$-Manhattan network defined on the same set of terminals when the number of directions in the unit ball
${\mathcal B}$ is increasing (the directions of the unit ball are indicated at the upper left corner of each figure).

\begin{figure}
 \begin{minipage}[b]{.55\linewidth}
  \centering \includegraphics[width=\linewidth]{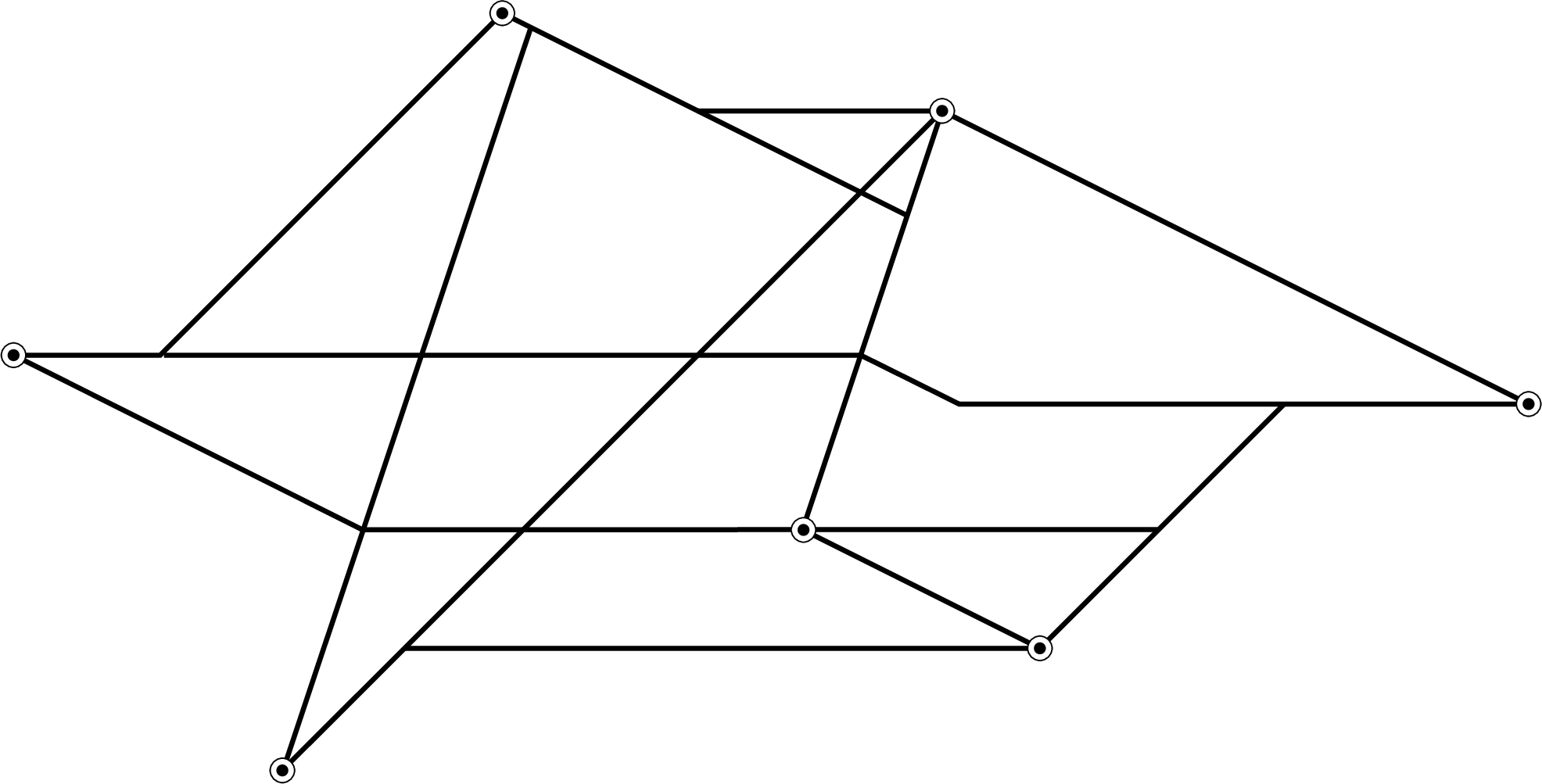}
  \caption{A ${\mathcal B}$-Manhattan network in the normed plane whose unit
ball
is depicted in  Fig. \ref{normed_plane} \label{bman}}
 \end{minipage} \hfill
 \begin{minipage}[b]{.40\linewidth}
  \centering \includegraphics[width=\linewidth]{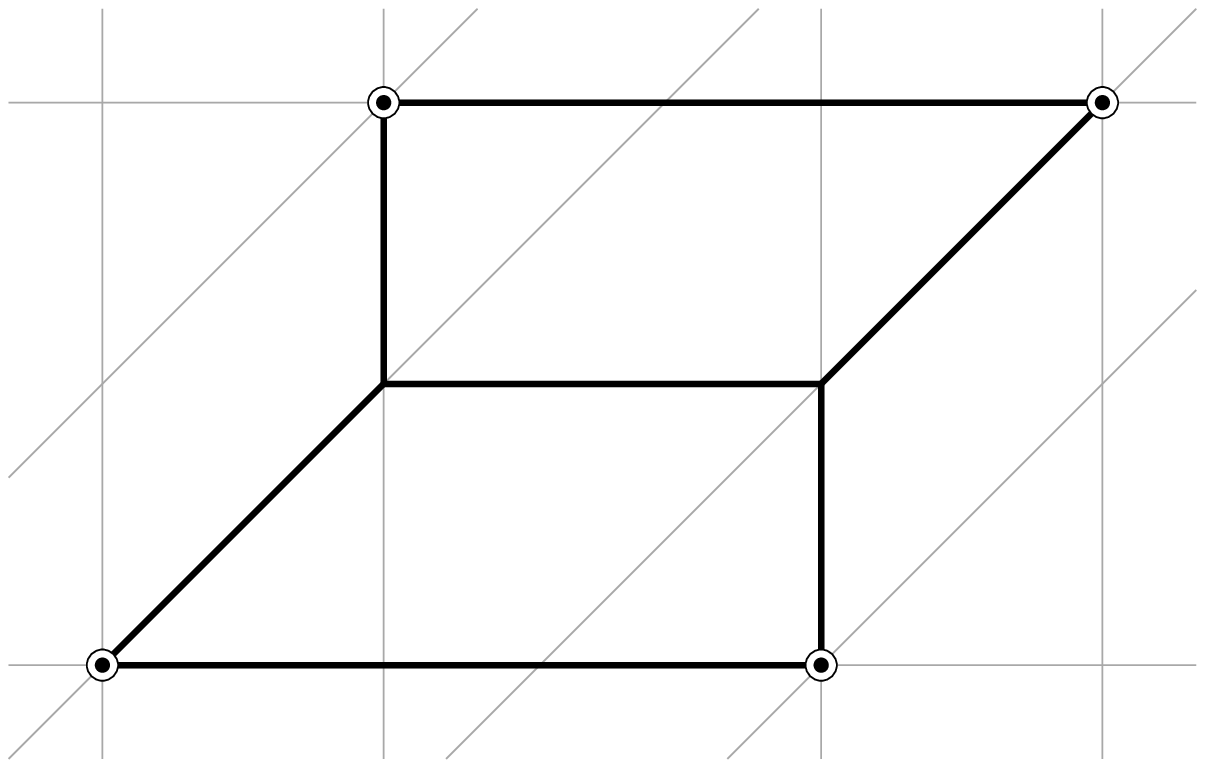}
  \caption{The unique optimal solution for this instance does not belong to the
grid $\Gamma$ (the unit ball ${\mathcal B}$ is a hexagon) \label{cexample}}
 \end{minipage}
\end{figure}

\subsection{Known results}

The minimum Manhattan  network problem has been introduced  by
Gudmundsson, Levcopoulos, and Narasimhan  \cite{GuLeNa}.
  Gudmundsson et al. \cite{GuLeNa}
proposed an $O(n^3)$-time 4-approximation algorithm, and an
$O(n\log{n})$-time 8-approximation algorithm. They also conjectured
that there exists a 2-approximation algorithm for this problem and
asked if this problem is NP-complete. Quite recently, Chin, Guo, and
Sun \cite{ChGuSu} solved this last open question from \cite{GuLeNa}
and established that indeed the minimum Manhattan network problem is
strongly  NP-complete. Kato, Imai, and Asano \cite{KaImAs} presented
a 2-approximation algorithm, however, their correctness proof is
incomplete (see \cite{BeWoWiSh}). Following \cite{KaImAs}, Benkert,
Wolff, Shirabe, and  Widmann \cite{BeWoWiSh} described  an
$O(n\log{n})$-time 3-approximation algorithm and presented a
mixed-integer programming formulation of the MMN problem.  Nouioua
\cite{Nou} and later  Fuchs and Schulze \cite{FuSch} presented two
simple $O(n\log{n})$-time 3-approximation algorithms. The first
correct 2-approximation algorithm (thus solving the first open
question from \cite{GuLeNa}) was presented by Chepoi, Nouioua, and
Vax\`es \cite{ChNouVa}. The algorithm is based on a strip-staircase
decomposition of the problem and uses a rounding method applied to
the optimal solution of the flow based linear program described in
\cite{BeWoWiSh}.  In his PhD thesis, Nouioua \cite{Nou}
described a $O(n\log{n})$-time 2-approximation algorithm based on
the primal-dual method from linear programming and the
strip-staircase decomposition.  In 2008, Guo, Sun, and Zhu
\cite{GuSuZh1,GuSuZh} presented two combinatorial factor 2
approximation algorithms, one with complexity $O(n^2)$ and another
with complexity $O(n\log{n})$ (see also the PhD thesis \cite{Sch} of
Schulze for yet another $O(n\log{n})$-time 2-approximation algorithm).
Finally, Seibert and Unger \cite{SeUn} announced a 1.5-approximation
algorithm, however the conference format of their paper does not
permit to understand the description of the algorithm and to check
its claimed performance guarantee (a counterexample that an
important intermediate step of their algorithm is incorrect was
given in \cite{FuSch,Sch}). Quite surprisingly, despite a
considerable number of prior work on minimum Manhattan network
problem,  no previous paper, to our knowledge, consider its
generalization to normed planes.

\begin{figure}
\begin{center}
\begin{tabular}{cccc}
\includegraphics[scale=0.35]{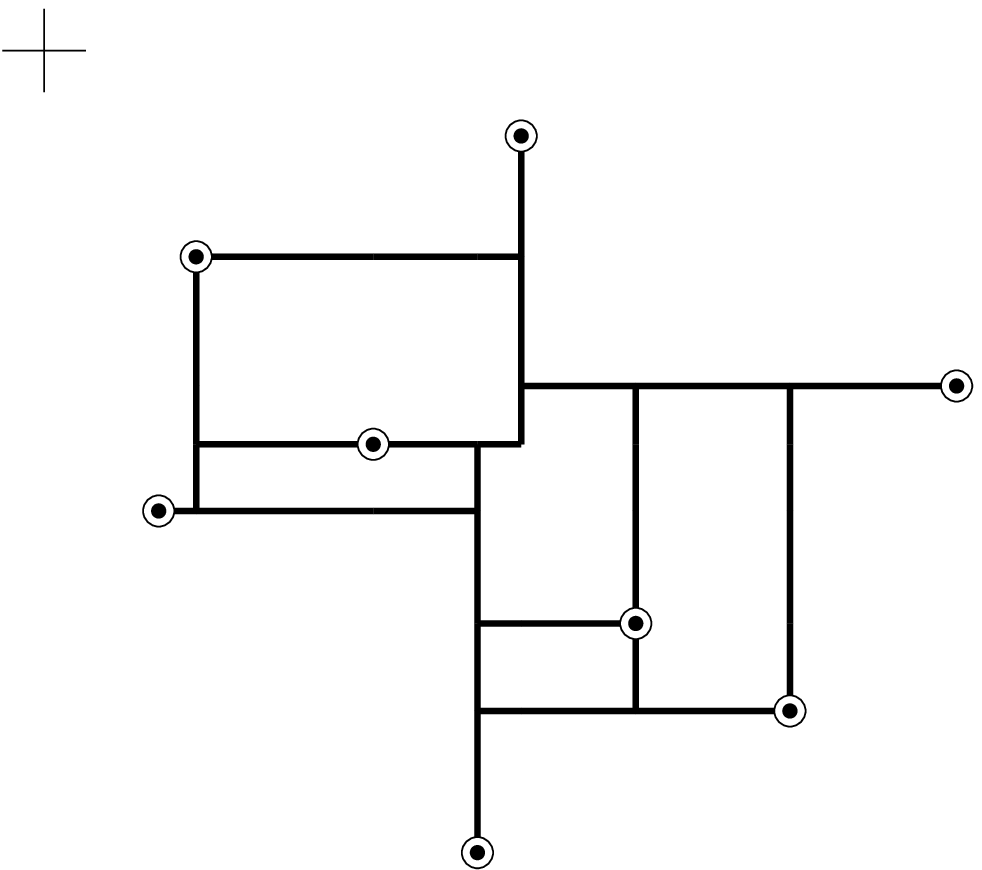} & \includegraphics[scale=0.35]{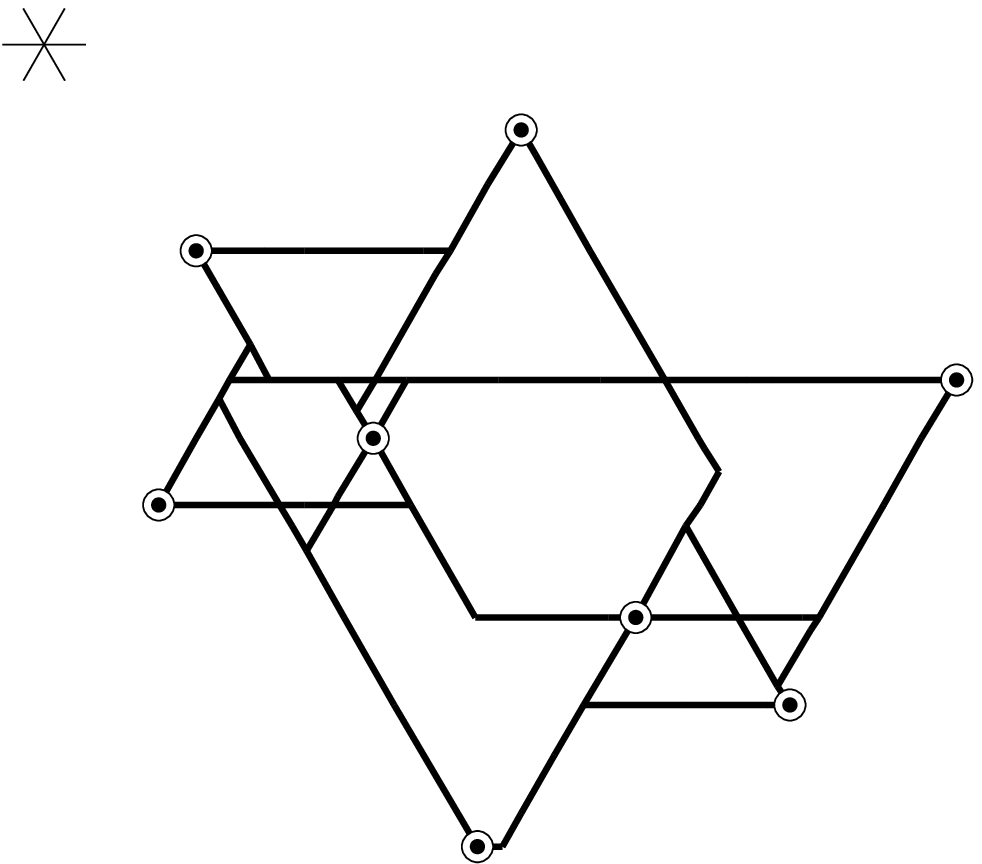} & \includegraphics[scale=0.35]{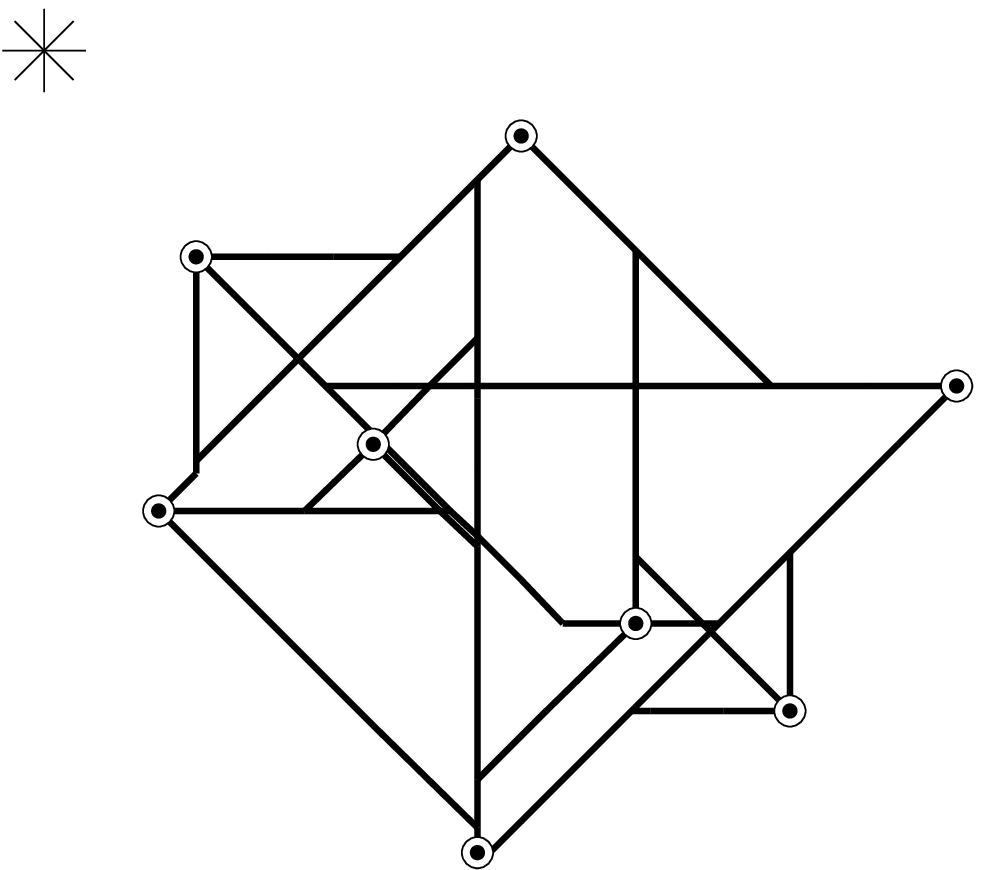} & \includegraphics[scale=0.35]{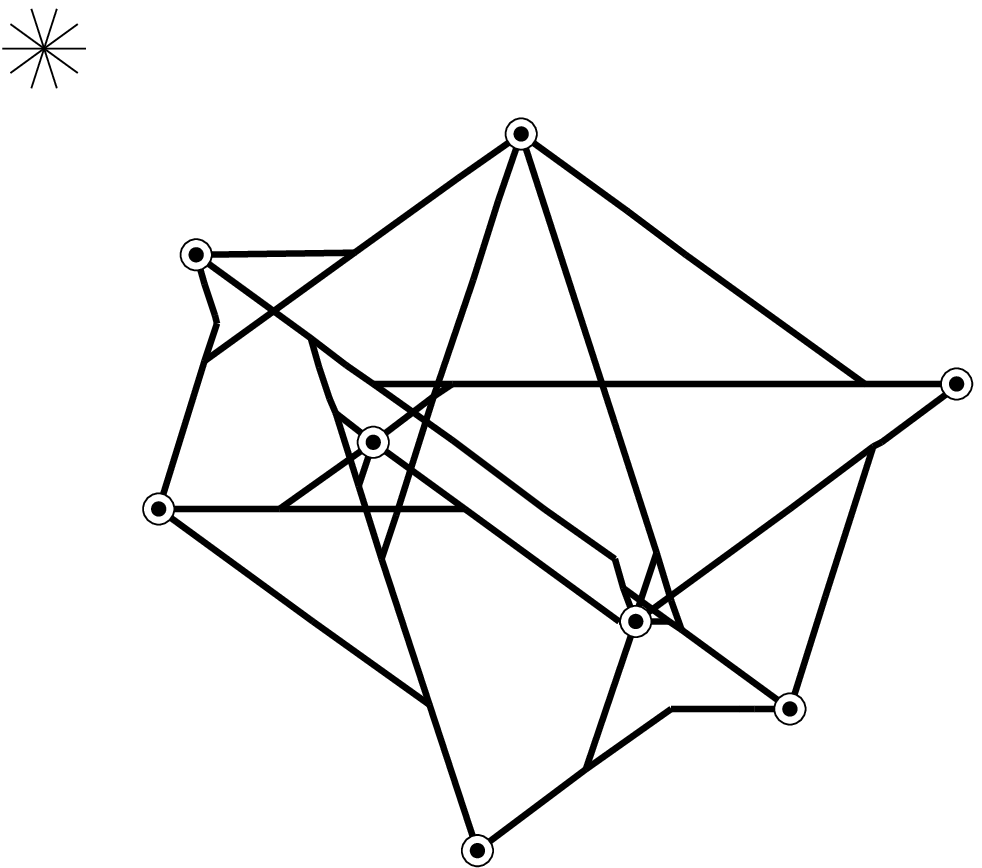}\\
\includegraphics[scale=0.35]{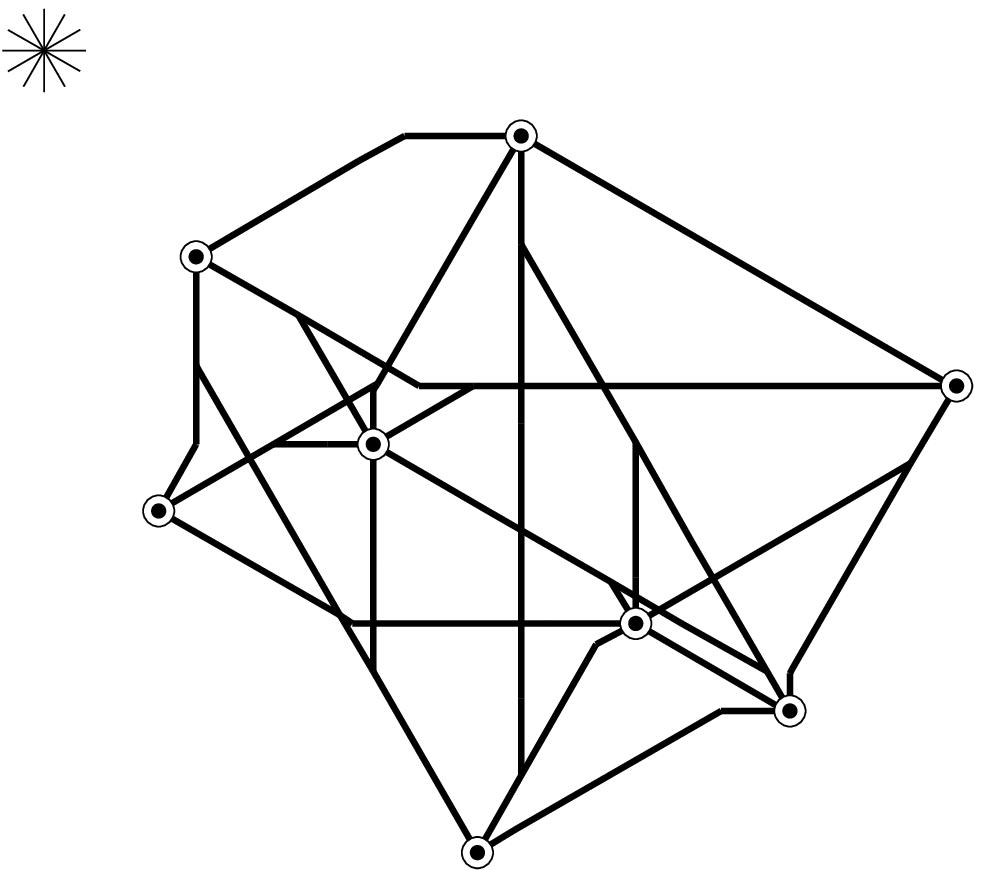} & \includegraphics[scale=0.35]{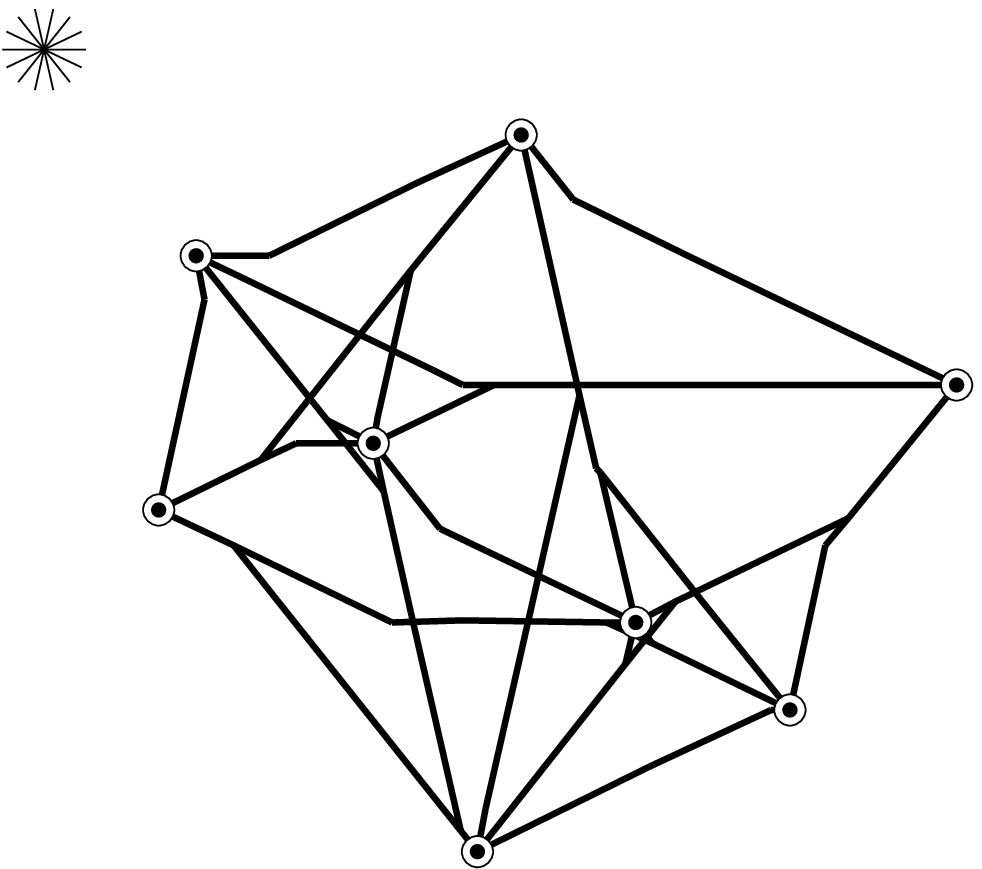} & \includegraphics[scale=0.35]{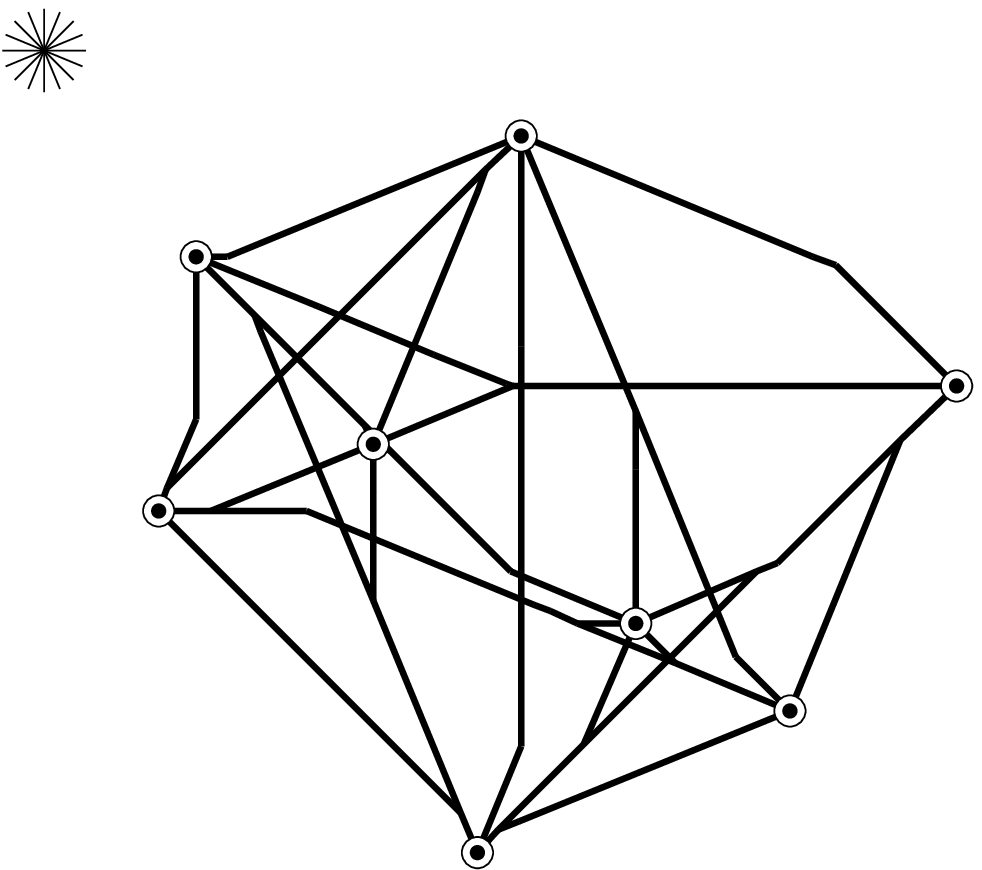} & \includegraphics[scale=0.35]{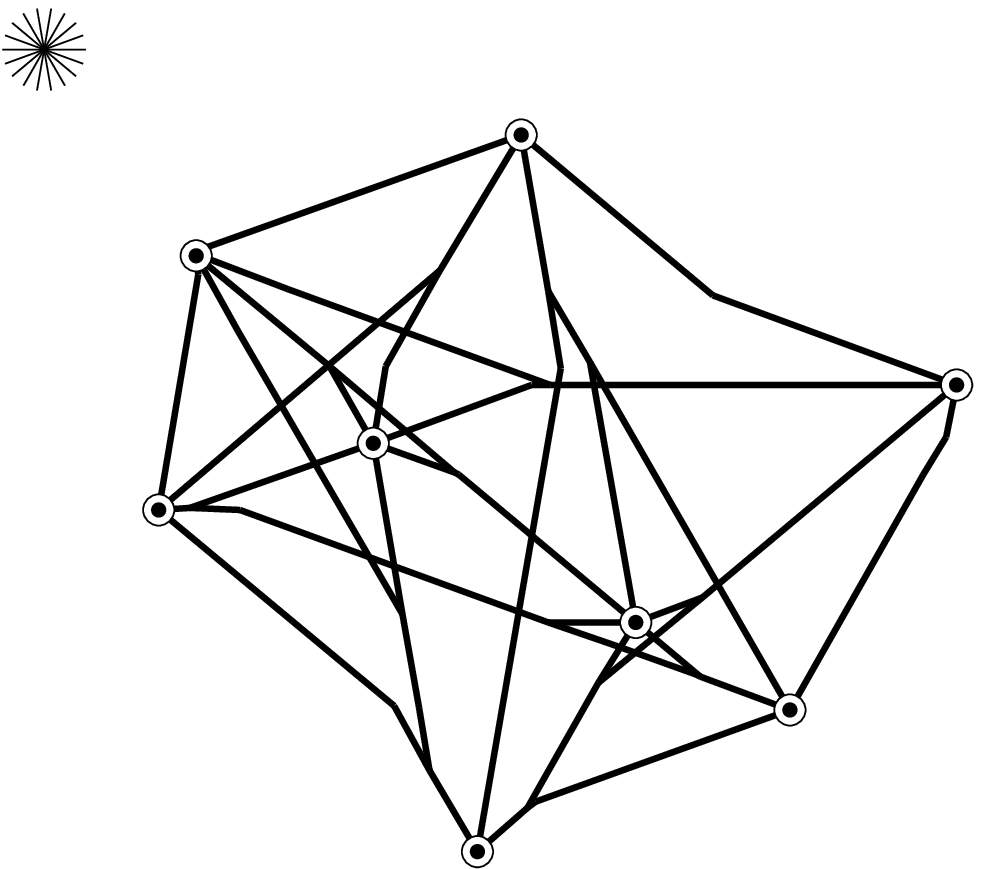}\\
\includegraphics[scale=0.35]{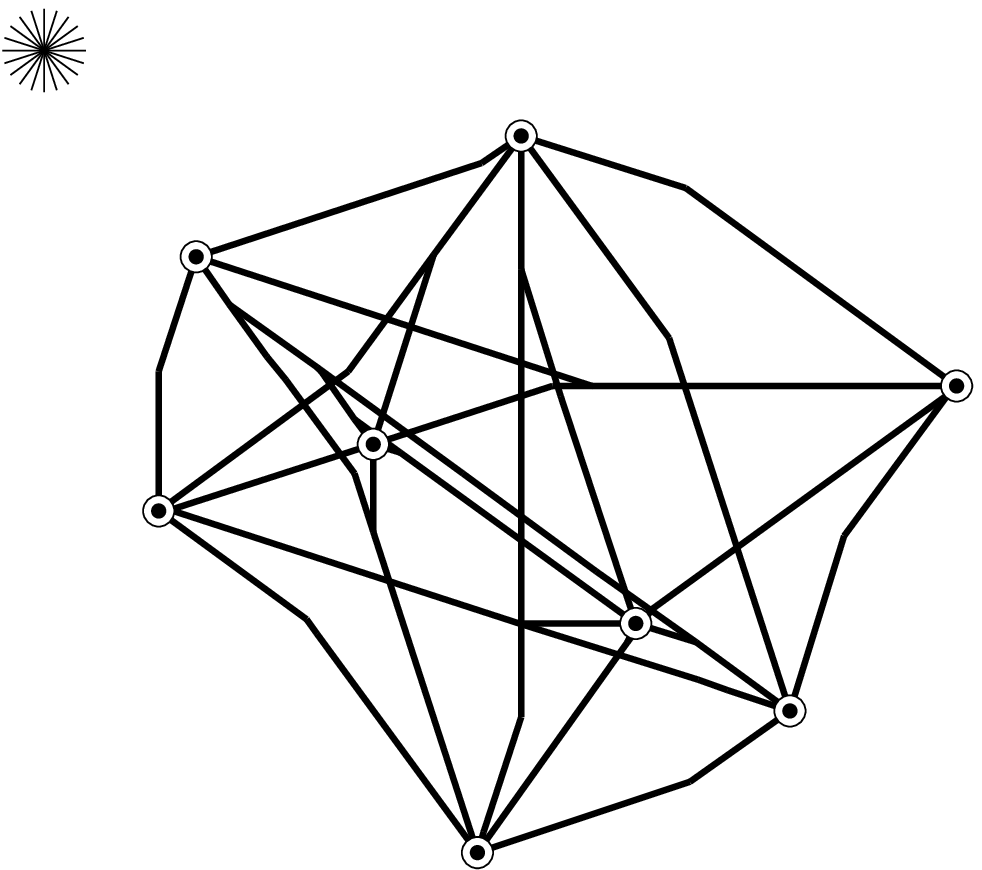} & \includegraphics[scale=0.35]{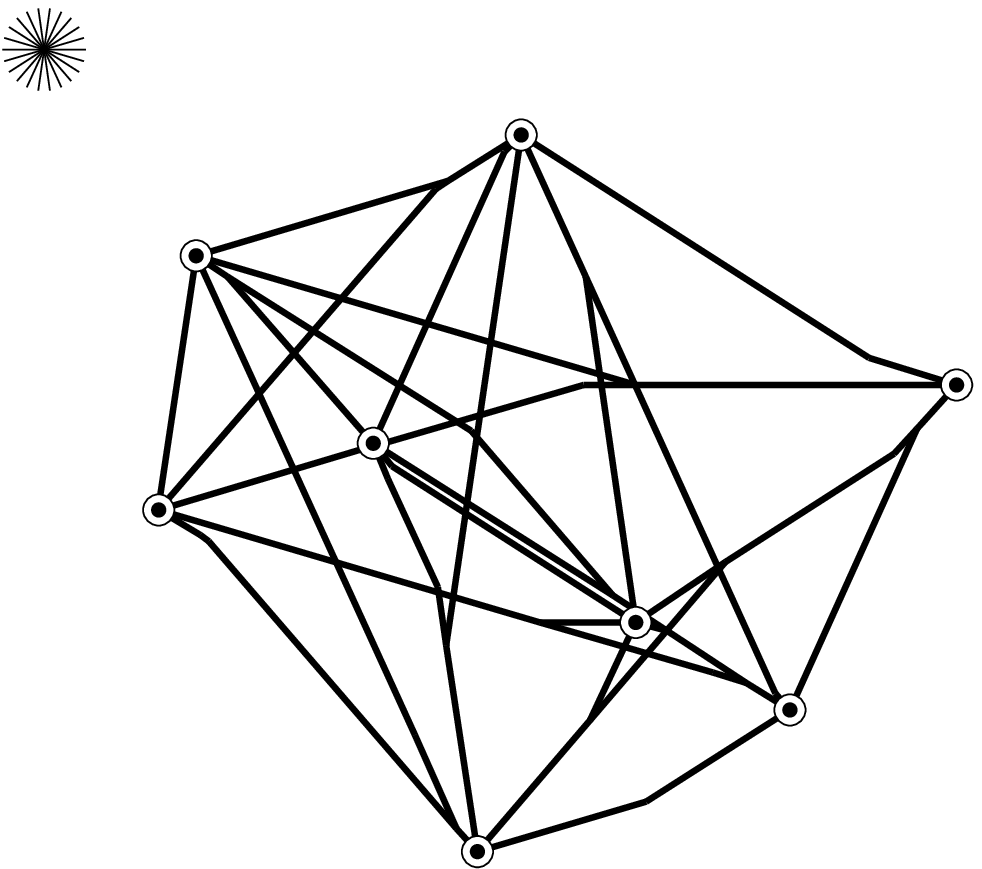} & \includegraphics[scale=0.35]{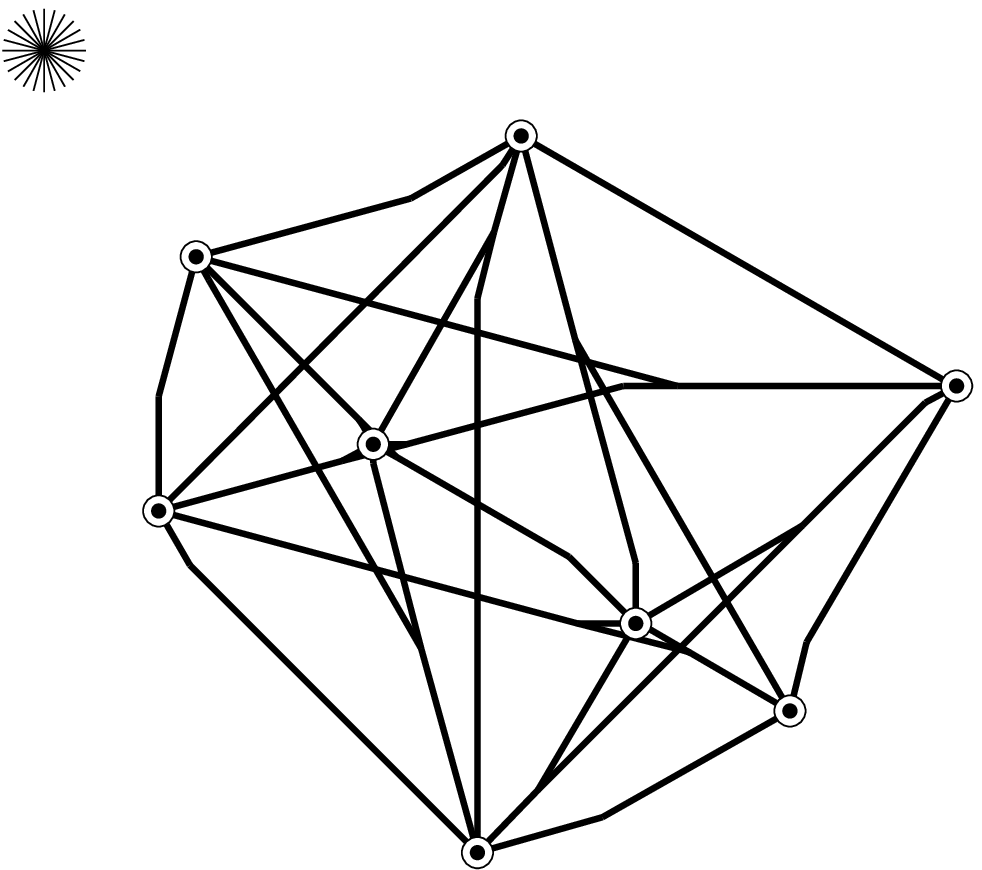} & \includegraphics[scale=0.35]{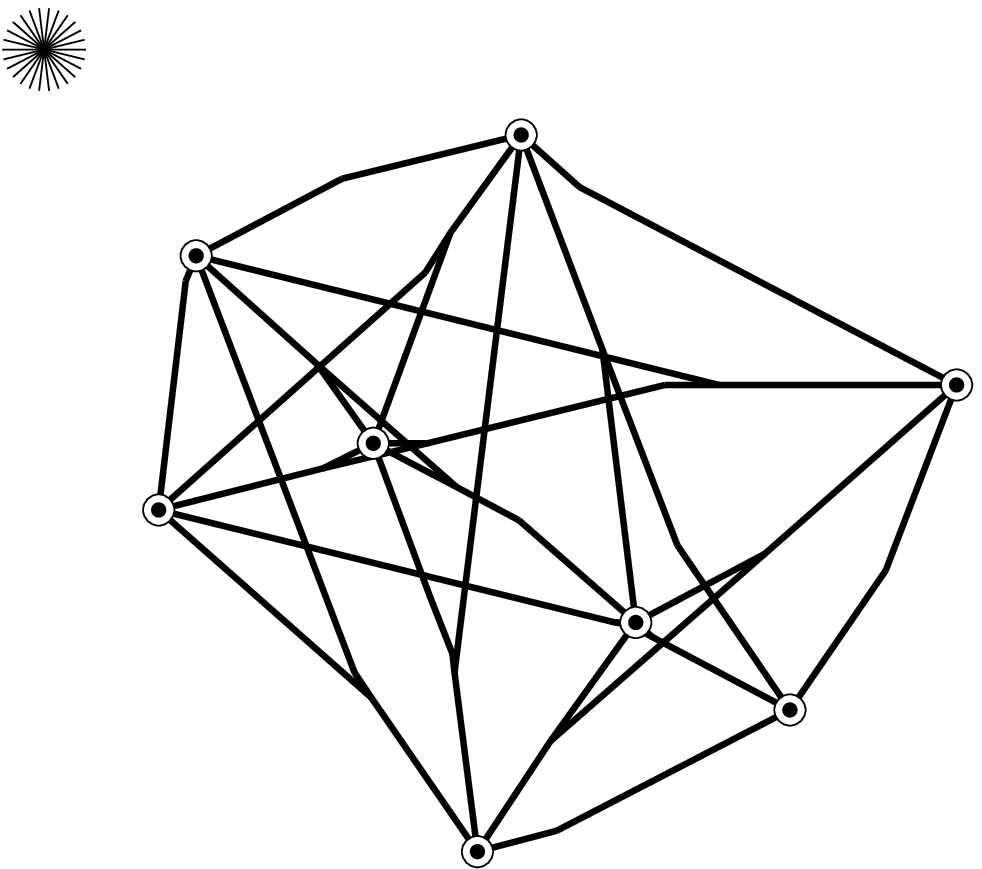}\\
\end{tabular}
\caption{``Morphing'' a minimum ${\mathcal
B}$-Manhattan network}
\label{bman_bis}
\end{center}
\end{figure}

%Minimum Manhattan network problem was also
%studied in three dimensions \cite{En,MuSeUn}

Gudmundsson et al. \cite{GuLeNa} introduced the minimum Manhattan
networks in connection with the construction of sparse geometric
spanners. Given a set $T$ of $n$ points in a normed plane  and a
real number $t\ge 1$, a geometric network $N$ is a {\it $t$-spanner}
for $T$ if for each pair of points ${\bf p},{\bf q}\in T,$ there
exists a $({\bf p},{\bf q})$-path in $N$ of length at most $t$ times
the distance $\|{\bf p}-{\bf q}\|$ between ${\bf p}$ and ${\bf q}.$
In the Euclidian plane and more generally, for normed planes with
round balls, the line segment is the unique shortest path between
two endpoints, and therefore the unique $1$-spanner of $T$ is the
complete graph on $T.$ On the other hand, if the unit ball of the
norm is a polygon, the points are connected by several shortest
${\mathcal B}$-paths, therefore the problem of finding the sparsest $1$-spanner
becomes non trivial. In this connection, minimum ${\mathcal
B}$-Manhattan networks are precisely the optimal $1$-spanners.
Sparse geometric spanners have applications in VLSI circuit design,
network design, distributed algorithms and other areas, see for
example the survey of \cite{Epp} and the book \cite{NaSm}. Lam,
Alexandersson, and Pachter \cite{LaAlPa} suggested to use minimum
Manhattan networks to design efficient search spaces for pair hidden
Markov model (PHMM) alignment algorithms.

Algorithms for solving different distance problems in normed spaces
with polygonal and polyhedral balls were proposed by Widmayer, Wu,
and Wang \cite{WiWuWa} (for more references and a systematic study
of such problems, see the book by Fink and Wood \cite{FiWo}). There
is also an extensive bibliography on facility location problems in
normed spaces with polyhedral balls, see for example
\cite{DuMi,ThWaWe}. Finally, the minimum Steiner tree problem in the
normed planes was a subject of intensive investigations, both from
structural and algorithmic points of view;
\cite{BrZa,BrThWeZa,DuGrLiWa} is just a short sample of papers on
the subject.

\section{Preliminaries}

\subsection{Definitions, notations, auxiliary results}

We continue  by setting some basic definitions, notations, and known results.
Let $\mathcal B$ be a zonotope of ${\mathbb R}^2$ with
$2m$ vertices ${\bf b}_0,\ldots {\bf b}_{2m-1}$ having its center of
symmetry at the origin of coordinates (see Fig. \ref{normed_plane}
for an example).  The  segment $s_k:={\bf b}_k{\bf b}_{{k+1}({\rm
mod} ~2m)}$ is a {\it side} of ${\mathcal B}$.  We will call the line
${\ell}_i$ passing via the points ${\bf b}_k$ and ${\bf b}_{k+m}$ an
{\it extremal line} of $\mathcal B$.  Two consecutive extremal lines
${\ell}_k$ and ${\ell}_{k+1}$ defines two opposite {\it elementary
$k$-cones} $C_k$ and $C_{k+m}=-C_k$ containing the sides $s_k$ and
$s_{k+m},$ respectively.  We extend this terminology, and call {\it
elementary $k$-cones with apex ${\bf v}$} the cones $C_k({\bf
v})=C_k+{\bf v}$ and $-C_k({\bf v})=C_{k+m}+{\bf v}$ obtained by
translating the cones $C_k$ and $C_{k+m}$ by the vector ${\bf v}.$
We will call a pair of consecutive lines $D_k=\{
{\ell}_k,{\ell}_{k+1}\}$ a {\it direction} of the normed
plane. Denote by ${\mathcal B}({\bf v},r)=r\cdot{\mathcal B}+{\bf
v}$ the ball of radius $r$ centered at the point ${\bf v}.$

Let $I({\bf p},{\bf q})=\{ {\bf z}\in {\mathbb R}^2: d({\bf p},{\bf
q})=d({\bf p},{\bf z})+d({\bf z},{\bf q})\}$ be the {\it interval}
between ${\bf p}$ and ${\bf q}.$ The inclusion ${\bf p}{\bf
q}\subseteq I({\bf p},{\bf q})$ holds for all normed spaces. If
${\mathcal B}$ is round, then ${\bf p}{\bf q}=I({\bf p},{\bf q}),$
i.e., the shortest path between ${\bf p}$ and ${\bf q}$ is unique.
Otherwise, $I({\bf p},{\bf q})$ may host a continuous set of
shortest paths.  The intervals $I({\bf p},{\bf q})$ in a normed
plane (and, more generally, in a normed space) can be constructed in
the following pretty way, described, for example, in the book
\cite{BoMaSo}. If ${\bf p}{\bf q}$ is a legal segment, then ${\bf
p}{\bf q}$ is the unique shortest path between ${\bf p}$ and ${\bf
q},$ whence $I({\bf p},{\bf q})={\bf p}{\bf q}.$ Otherwise, set
$r=d({\bf p},{\bf q}).$ Let $s'_k$ be the side of the ball
${\mathcal B}({\bf p},r)$ containing the point $\bf q$ and let
$s''_j$ be the side of the ball ${\mathcal B}({\bf q},r)$ containing
the point $\bf p$. Notice that these sides are well-defined,
otherwise  $\bf q$ is a vertex of ${\mathcal B}({\bf p},r)$ and {\bf
p}{\bf q} is a legal segment. The segments  $s'_k$ and $s''_j$ are
parallel, thus $|k-j|=m,$ say $k\le m$ and $j=k+m.$ Then $I({\bf
p},{\bf q})$ is the intersection of the elementary cones  $C_k({\bf
p})$ and $C_{k+m}({\bf q})=-C_{k}({\bf q})$  (see Fig.
\ref{normed_plane} for an illustration):

\begin{lemma} \label{interval} \emph{\cite{BoMaSo}} $I({\bf p},{\bf q})=C_k({\bf
p})\cap (-C_{k}({\bf q})).$
\end{lemma}

\begin{center}
\begin{figure}
\hspace*{5mm}
\includegraphics[scale=0.4]{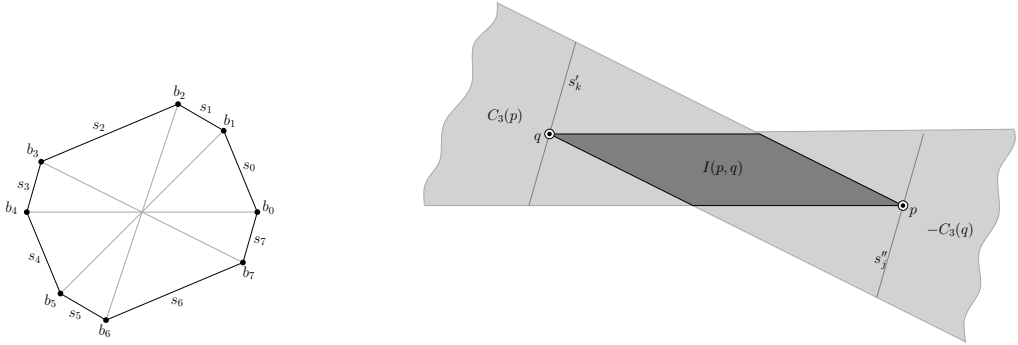}
\caption{A unit ball ${\mathcal B}$ and an interval
$I({\bf p},{\bf q})$} \label{normed_plane}
\end{figure}
\end{center}

An immediate consequence of this result is the following characterization of
shortest ${\mathcal B}$-paths between two points $\bf p$ and $\bf q$.

\begin{lemma} \label{shortest_legal_path} If  ${\bf p}{\bf q}$ is a legal
segment, then
${\bf p}{\bf q}$ is the unique shortest  ${\mathcal B}$-path. Otherwise, if
$I({\bf p},{\bf q})=C_k({\bf p})\cap (-C_{k}({\bf q})),$ then any
shortest  ${\mathcal B}$-path $\pi({\bf p},{\bf q})$ between  ${\bf p}$ and
${\bf q}$ has only $k$-segments and $(k+1)$-segments as edges.
Moreover, $\pi({\bf p},{\bf q})$ is a shortest ${\mathcal B}$-path if and
only if it is monotone with respect to ${\ell}_k$ and
${\ell}_{k+1},$ i.e., the intersection of $\pi({\bf p},{\bf q})$
with any line $\ell$ parallel to the lines ${\ell}_k,{\ell}_{k+1}$
is empty, a point, or a (legal) segment.
\end{lemma}

\begin{proof} The first statement immediately follows from Lemma
\ref{interval}. Suppose that ${\bf p}{\bf q}$ is not a legal segment and $I({\bf
p},{\bf q})=C_k({\bf p})\cap C_{k+m}({\bf q}).$
Let ${\bf u}{\bf v}$ be the first edge on a shortest path $\pi({\bf p},{\bf
q})$  from ${\bf p}$ to  ${\bf q}$ which is neither a $k$-segment nor
a $(k+1)$-segment. Since ${\bf u}\in I({\bf p},{\bf q})=C_k({\bf p})\cap
C_{k+m}({\bf q}),$ the point ${\bf q}$ belongs to the cone
$C_k({\bf u})$ and the point ${\bf u}$ belongs to the cone $C_{k+m}({\bf
q}),$ whence $I({\bf u},{\bf q})=C_k({\bf u})\cap C_{k+m}({\bf q}).$
Obviously, the point ${\bf v}$ belongs to $I({\bf u},{\bf q}).$ However, by the
choice of the segment ${\bf u}{\bf v}$ and the fact that ${\ell}_k$ and
${\ell}_{k+1}$
are consecutive lines that forms a direction, the point ${\bf v}$
cannot belong $C_k({\bf u}),$ a contradiction.
This shows that any shortest legal path $\pi({\bf p},{\bf q})$ between
${\bf p}$ and ${\bf q}$ has only $k$- and $(k+1)$-segments as edges.
Additionally,
the intersection of $\pi({\bf p},{\bf q})$ with any line
$\ell$ parallel to ${\ell}_k$ or ${\ell}_{k+1}$ is empty,
a point, or a (legal) segment. Indeed,
pick any two points in this intersection. Since the legal segment
defined by these points is the unique shortest path between them, it must also
belong to
the intersection of $\pi({\bf p},{\bf q})$ with $\ell$.  Conversely,
consider  a monotone path $\pi({\bf p},{\bf q})$ between ${\bf p}$ and ${\bf
q},$
namely suppose that the  intersection
of $\pi({\bf p},{\bf q})$ with any line $\ell$ parallel to
the lines ${\ell}_k$ or ${\ell}_{k+1}$ is empty, a point, or a (legal) segment).
We proceed by induction on the number of edges of  $\pi({\bf p},{\bf q}).$
The monotonicity of $\pi({\bf p},{\bf q})$ implies that $\pi({\bf p},{\bf q})$
lies entirely in the interval $I({\bf
p},{\bf q}).$ In particular, the neighbor ${\bf u}$ of ${\bf p}$ in $\pi({\bf
p},{\bf q})$ belongs to $I({\bf p},{\bf q}).$ The subpath $\pi({\bf u},{\bf q})$
of $\pi({\bf p},{\bf q})$ between ${\bf u}$ and ${\bf q}$ is monotone, therefore
by induction assumption,   $\pi({\bf u},{\bf q})$ is a shortest path between
${\bf u}$ and ${\bf q}.$ Since ${\bf p}{\bf u}$ is a legal segment and ${\bf
u}\in I({\bf p},{\bf q}),$
we immediately conclude that $\pi({\bf p},{\bf q})$ is also a shortest path
between ${\bf p}$ and ${\bf q}.$  $\Box$
\end{proof}

We continue with some notions and notations about the ${\mathcal
B}$-MMN problem. Denote by OPT$(T)$ the length of a minimum
${\mathcal B}$-Manhattan network for a set of terminals $T$. For a
direction $D_k=\{\ell_k,\ell_{k+1}\},$ denote by $F_k$ the set of
all pairs $\{i,j\}$ (or pairs of terminals $\{ {\bf t}_i,{\bf
t}_j\}$) such that any shortest ${\mathcal B}$-path between  ${\bf t}_i$ and
${\bf t}_j$ uses only $k$-segments and $(k+1)$-segments.
Equivalently, by Lemma \ref{shortest_legal_path}, $F_k$ consists of all
pairs of terminals which belong
to two opposite elementary cones $C_k({\bf v})$ and $-C_k({\bf v})$
with common apex.  For each direction $D_k$ and the set of pairs
$F_k,$ we formulate an auxiliary problem which we call {\it Minimum
1-Directional Manhattan Network problem} (or {\it 1-DMMN$(F_k)$ problem}):
find a network $N^{opt}_k(T)$ of minimum possible length such that every edge
of $N^{opt}_k(T)$ is an $k$-segment or an $(k+1)$-segment and any pair $\{
{\bf t}_i,{\bf t}_j\}$ of $F_k$ is connected in  $N^{opt}_k(T)$ by a shortest
${\mathcal B}$-path. We denote its length by ${\rm OPT}_k(T).$ We continue by
adapting to 1-DMMN the notion of a generating set introduced in \cite{KaImAs}
for MMN problem: a {\it generating set} for $F_k$ is a subset $F$ of $F_k$ with
the
property that a $\mathcal B$-Manhattan  network containing shortest
$\mathcal B$-paths for all pairs in $F$ is a 1-Directional Manhattan
network for $F_k.$

\subsection{Our approach}

Let  $N^*(T)$ be a minimum  ${\mathcal B}$-Manhattan network, i.e.,
a ${\mathcal B}$-Manhattan network of total length $l(N^*(T))={\rm
OPT}(T).$ For  each direction $D_k,$ let $N^*_k(T)$ be the set of
$k$-segments and $(k+1)$-segments of $N^*(T).$ The network
$N^*_k(T)$ is an admissible solution for 1-DMMN$(F_k),$
thus the length $l(N^*_k(T))$ of $N^*_k(T)$ is at least ${\rm
OPT}_k(T).$ Any $k$-segment of $N^*(T)$ belongs to two one-directional networks
$N^*_{k}(T)$ and $N^*_{k-1}(T).$ Vice-versa, if $N_k(T),
k=0,\ldots,m-1,$ are admissible solutions for the 1-DMMN$(F_k)$ problems,
since $\bigcup_{k=0}^{m-1} F_k=T\times T,$ the network
$N(T)=\bigcup_{k=0}^{m-1} N_k(T)$ is a ${\mathcal B}$-Manhattan
network. Moreover, if  each $N_k(T)$ is an $\alpha$-approximation
for respective 1-DMMN problem, then the network $N(T)$ is a
$2\alpha$-approximation for the minimum  ${\mathcal B}$-Manhattan
network problem. Therefore, to obtain a factor 2.5-approximation for
$\mathcal B$-MMN, we need to provide a 1.25-approximation for the
1-DMMN problem.  The remaining part of our paper describe such a
combinatorial algorithm.  The  1-DMMN problem is easier and less
restricted than the ${\mathcal B}$-MMN problem because we have to connect with
shortest paths only the pairs of terminals of the set $F_k$
corresponding to one direction $D_k$, while in case of the MMN
problem the set $T\times T$ of all pairs is partitioned into two
sets corresponding to the two directions of the $l_1$-plane. For our
purposes, we will adapt the strip-staircase decomposition of
\cite{ChNouVa}, by considering only the strips and the staircases
which ``are oriented in direction $D_k$''.

\section{One-directional strips and staircases}

In the next two sections, we assume that $D_k=\{ l_k,l_{k+1}\}$ is a
fixed but arbitrary direction of the normed plane. We recall the
definitions of vertical and horizontal strips and staircases
introduced in \cite{ChNouVa}.  Then we consider only those of them
which correspond to pairs of terminals from the set $F_k,$ which we
call one-directional strips and staircases. We formulate several
properties of one-directional strips and staircases  and we prove
those of them which do not hold for usual strips and staircases.

Denote by $L_k$ and $L_{k+1}$ the set of all lines
passing via the terminals of $T$ and parallel to the extremal lines $\ell_k$ and
$\ell_{k+1},$ respectively.
Let $\Gamma_k$ be the grid defined by the lines of $L_k$ and $L_{k+1}.$ The
following lemma can be proved in the same way as for rectilinear Steiner trees
or Manhattan networks
(quite surprisingly, this is not longer true for the ${\mathcal B}$-MMN problem:
Fig. \ref{cexample} presents an instance of ${\mathcal B}$-MMN
for which the unique optimal solution does not belong to the grid
$\Gamma:=\bigcup_{k=0}^{m-1} \Gamma_k$):

\begin{lemma} \label{grid} There exists a minimum 1-Directional Manhattan
Network for $F_k$ contained in the grid $\Gamma_k$.
\end{lemma}

For two terminals ${\bf t}_i,{\bf t}_{i'},$ set $R_{i,i'}:=I({\bf t}_i,{\bf t}_{i'}).$
%A pair  $\{ {\bf t}_i,{\bf t}_j\}$ is {\it empty} if $R_{i,j}\cap T=\{ {\bf
%t}_i,{\bf t}_j\}$.
A pair  ${\bf t}_i,{\bf t}_{i'}$ defines a {\it $k$-strip} $R_{i,i'}$ if
either (i) (degenerated strip) ${\bf t}_i$ and ${\bf t}_{i'}$ are
consecutive terminals belonging to the same line of $L_k$ or (ii)
${\bf t}_i$ and ${\bf t}_{i'}$ belong to two consecutive lines of $L_k$
and the intersection of $R_{i,i'}$ with any degenerated $k$-strip is
either empty or one of the terminals ${\bf t}_i$ or ${\bf t}_{i'};$ see
Fig.~6 of \cite{ChNouVa}The two $k$-segments of $R_{i,i'}$ are called
the {\it sides} of $R_{i,i'}.$ The {\it $(k+1)$-strips} and their
sides are defined analogously (with respect to $L_{k+1}$). With some
abuse of language, we will call the $k$-strips {\it horizontal} and
the $(k+1)$-strips {\it vertical}.  If a pair $\{ {\bf t}_i,{\bf
t}_{i'}\}$ defining a horizontal or a vertical strip $R_{i,i'}$ belongs
to the set $F_k,$ then we say that $R_{i,i'}$ is a {\it
one-directional strip} or a {\it 1-strip}, for short. Denote by
$F'_k$ the set of all pairs of $F_k$  defining one-directional
strips.

%Let $F^'_k$ and $F^'_k}$ denote the pairs of  $F'_k$ defining horizontal and
%vertical strips, respectively.

\begin{lemma} \label{strips} If $R_{i,i'}$ and $R_{j,j'}$ are two horizontal
1-strips or two vertical 1-strips, then  $R_{i,i'}\cap R_{j,j'}=\emptyset$ if
$\{ i,i'\}\cap \{ j,j'\}=\emptyset$ and
$R_{i,i'}\cap R_{j,j'}=\{ {\bf t}_i\}$ if $\{ i,i'\}\cap \{ j,j'\}=\{ i\}.$
\end{lemma}

\begin{proof} From the definition follows that if $R_{i,i'}$ and $R_{j,j'}$ are
both degenerated or one is
degenerated  and another one not, then they are either disjoint or
intersect in a single terminal. If $R_{i,i'}$ and $R_{j,j'}$ are
both non-degenerated and intersect, then from the definition
immediately follows that  the intersection is one point or a segment
belonging to their sides. However, if $R_{i,i'}$ and $R_{j,j'}$
intersects in a segment, then  one can easily see that at least one
of $R_{i,i'}$ and $R_{j,j'}$ cannot be a 1-strip. $\Box$
\end{proof}

We say that a vertical 1-strip $R_{i,i'}$ and a horizontal 1-strip
$R_{j,j'}$ (degenerated or not) form a {\it crossing configuration}
if they intersect (and therefore cross each other).

\begin{lemma} \label{crossing} If $R_{i,i'}$ and $R_{j,j'}$ form a
crossing configuration, then from the shortest ${\mathcal B}$-paths between
${\bf t}_i$
and ${\bf t}_{i'}$ and between ${\bf t}_j$ and ${\bf t}_{j'}$ one  can derive
shortest ${\mathcal B}$-paths
connecting ${\bf t}_i, {\bf t}_{j'}$ and ${\bf t}_{i'}, {\bf t}_j,$
respectively.
\end{lemma}

\begin{figure}
 \begin{center}
  \includegraphics[scale=0.43]{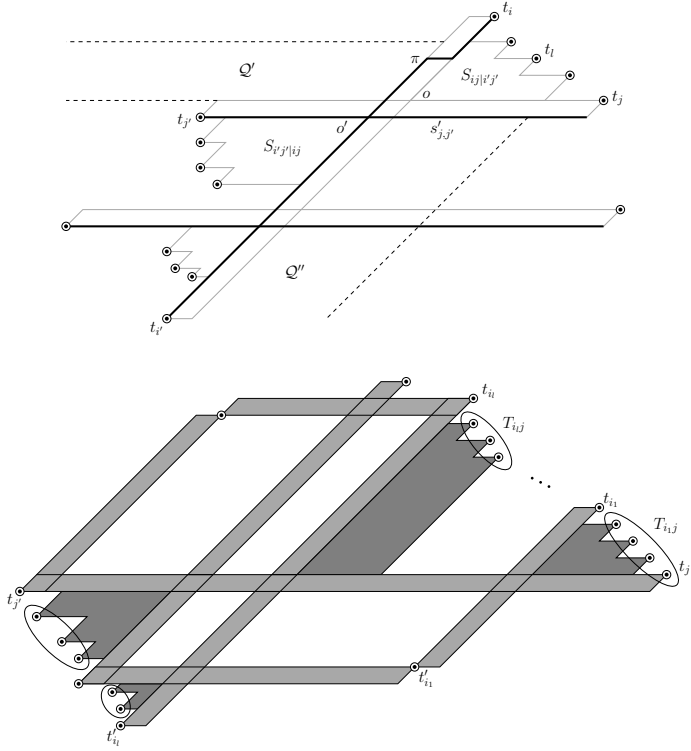}
  \caption{Strips, staircases, and completion}
\label{fig_strips_staircases_completion}
 \end{center}
\end{figure}

For a crossing configuration defined by the 1-strips
$R_{i,i'},R_{j,j'},$ denote by ${\bf o}$ and ${\bf o}'$ the two
opposite corners of the parallelogram $R_{i,i'}\cap R_{j,j'},$ such
that the cones $C_k({\bf o})$ and $-C_k({\bf o}')$ do not intersect
the interiors of $R_{i,i'}$ and $R_{j,j'}.$  Additionally, suppose
without loss of generality, that ${\bf t}_i$ and ${\bf t}_j$ belong
to the cone $C_k({\bf o}),$ while ${\bf t}_{i'}$ and ${\bf t}_{j'}$
belong to the cone $-C_k({\bf o}').$ Denote by $T_{i,j}$ the set of
all terminals  ${\bf t}_l\in (T\setminus \{ {\bf t}_i,{\bf t}_j\})
\cap C_k({\bf o})$ such that $(-C_k({\bf t}_l)) \setminus (-C_k({\bf
o}))$ does not contain any terminal except $t_l.$  Denote by ${\cal
S}_{i,j|i',j'}$ the region of $C_k({\bf o})$ which is the union of
the intervals  $I({\bf t}_l,{\bf o}), {\bf t}_l\in T_{i,j},$ and
call this polygon an {\it one-directional staircase} or a {\it
1-staircase}, for short; see Fig.
\ref{fig_strips_staircases_completion} and Figures 7,8 of
\cite{ChNouVa} for an illustration. Note that ${\cal S}_{i,j|i',j'}$
is bounded by the 1-strips $R_{i,i'}$ and $R_{j,j'}$ and a legal
path between ${\bf t}_i$ and ${\bf t}_j$ passing via all terminals
of $T_{i,j}$ and consisting of $k$-segments and $(k+1)$-segments.
The point $o$ is called the {\it origin}  and $R_{i,i'}$ and
$R_{j,j'}$ are called the {\it basis} of this staircase. Since
$I({\bf t}_l,{\bf o})\subset (-C_k({\bf t}_l))\setminus (-C_k({\bf
o}))$ for all ${\bf t}_l\in T_{i,j},$ $I({\bf t}_l,{\bf o})\cap T=\{
{\bf t}_l\}$ and therefore ${\cal S}_{i,j|i',j'}\cap T=T_{i,j}.$ For
the same reason, there are no terminals of $T$ located in the
regions ${\mathcal Q}'$ and ${\mathcal Q}''$ depicted in Fig.
\ref{fig_strips_staircases_completion} (${\mathcal Q}'$ is the
region comprised between the leftmost side of $R_{i,i'},$ the
highest side of   $R_{j,j'},$ and the line of $L_k$ passing via the
highest terminal of $T_{i,j},$ while ${\mathcal Q}''$ is the region
comprised between the rightmost side of $R_{i,i'},$ the lowest side
of   $R_{j,j'},$ and the line of $L_{k+1}$  passing via the
rightmost terminal of $T_{i,j}$). Analogously one can define the set
$T_{i',j'}$ and the staircase ${\cal S}_{i',j'|i,j}$ with origin
${\bf o}'$ and basis $R_{i,i'}$ and $R_{j,j'}.$

\begin{lemma} \label{strip_staircase} If a 1-strip $R_{l,l'}$
intersects a 1-staircase ${\cal
S}_{i',j'|i,j}$ and  $R_{l,l'}$ is different from the 1-strips
$R_{i,i'}$ and $R_{j,j'},$ then  $R_{l,l'}\cap {\cal
S}_{i',j'|i,j}$ is a single terminal.
\end{lemma}

\begin{proof} If a 1-strip  $R_{l,l'}$ traverses a
staircase ${\cal S}_{i',j'|i,j}$, then  one of the terminals ${\bf
t}_l,{\bf t}_{l'}$  must be located in one of the regions ${\mathcal
Q}'$ and ${\mathcal Q}'',$ which is impossible because $({\mathcal
Q}'\cup {\mathcal Q}'')\cap T=\emptyset.$ Thus, if $R_{l,l'}$ and
${\cal S}_{i',j'|i,j}$ intersect more than in one point, then they
intersect in a segment $s$ which belongs to one side of $R_{l,l'}$
and to the boundary of ${\cal S}_{i',j'|i,j}.$ If say the 1-strip
$R_{l,l'}$ is horizontal, then necessarily  $s$ is a part of the
lowest side of $R_{l,l'}$ and of the highest horizontal side of
${\cal S}_{i',j'|i,j}.$ Let ${\bf t}$ be the highest terminal of
$T_{i,j}.$ Then either ${\bf t}$ belongs to $R_{l,l'}$ and is
different from ${\bf t}_l,{\bf t}_{l'},$ contrary to the assumption
that $R_{l,l'}$ is a strip, or ${\bf t}$ together with the lowest
terminal ${\bf t}_{l'}$ of $R_{l,l'}$ define a degenerated strip
with ${\bf t}_{l'}$ belonging to ${\mathcal Q}',$ contrary to the
assumption that ${\mathcal Q}'\cap T=\emptyset$. $\Box$
\end{proof}

\begin{lemma} \label{staircase_staircase} Two 1-staircases either are
disjoint or intersect only in common terminals.
\end{lemma}

\begin{proof} From the definition of a staircase
follows that the interiors of two staircases are disjoint (for a
short formal proof of this see \cite{ChNouVa}). Therefore two
staircases may intersect only on the boundary.  In this case, the
intersection is either a subset of terminals of both staircases or a
single edge. In the second case, one of the two staircases
necessarily is not a 1-staircase with respect to the chosen
direction. $\Box$
\end{proof}

Let $F''_k$ be the set of all pairs $\{{\bf t}_{j'},{\bf t}_l\}$
such that there exists a 1-staircase ${\cal S}_{i,j|i',j'}$ with
${\bf t}_l$ belonging to the set $T_{i,j}.$ The proof of the
following essential result is identical to the proof of Lemma 3.2 of
\cite{ChNouVa} and therefore is omitted.

\begin{lemma} \label{generating_set} $F:=F'_k\cup F''_k$ is a generating set for $F_k$.
\end{lemma}

%\begin{figure}
%\begin{center}
%\input{Staircase.tex}
%\caption{Staircase ${\cal S}_{i,j|i',j'}$}
%\end{center}
%\end{figure}

\section{The algorithm}

We continue with the description of our factor 1.25 approximation
algorithm for  1-DMMN problem. Let $F^h_k$ and $F^v_k$ denote the
pairs of  $F'_k$ defining horizontal and vertical 1-strips,
respectively. Let $S_1^h$ and $S^h_2$ be the networks consisting of
lower sides and respectively upper sides of the horizontal 1-strips
of $F^h_k.$  Analogously, let  $S_1^v$ and $S^v_2$ be the networks
consisting of rightmost sides and respectively leftmost sides of the
vertical 1-strips of $F^v_k.$ The algorithm completes optimally each
of the networks $S_1^h,S_2^h,S_1^v,$ and $S_2^v,$ and from the set
of four completions $N_1^h,N_2^h,N_1^v,N^v_2,$  the algorithm
returns the shortest one, which we will denote by $N_k(T)$ (in this
respect, our algorithm has some similarity with the approach of
Benkert et al. \cite{BeWoWiSh}). We will describe now the optimal
completion $N^h_1$ for the network $S^h_1,$ the three other networks
are completed in the same way (up to symmetry).

An {\it optimal completion} of $S^h_1$ is a subnetwork $N^h_1$ of $\Gamma_k$
extending $S^h_1$ ($S^h_1\subseteq N^h_1$) of smallest total length such that
any pair of terminals of $F$ can be connected in  $N^h_1$ by a shortest path.
By Lemma \ref{generating_set}, to solve the completion problem for $S^h_1$, it
suffices to (i)  select a shortest path $\pi({\bf t}_i,{\bf t}_{i'})$ of
$\Gamma_k$ between each pair ${\bf t}_i,{\bf t}_{i'}$ defining a vertical
1-strip $R_{i,i'}$, (ii) for each horizontal 1-strip $R_{j,j'}$  find a shortest
path $\pi({\bf t}_j,{\bf t}_{j'})$ between ${\bf t}_j$ and ${\bf t}_{j'}$
subject to the condition that the lowest side $s'_{j,j'}$ of $R_{j,j'}$ is
already available,   (iii) for each staircase $S_{i,j|i',j'}$ whose sides are
$R_{i,i'}$ and $R_{j,j'}$ select shortest paths from the terminals of $T_{i,j}$
to the terminal ${\bf t}_{j'}$ subject to the condition that the lowest side
$s'_{j,j'}$ of $R_{j,j'}$ is already available. We need to minimize the total
length of the resulting network $N_1^h$ over all vertical 1-strips and all
1-staircases. To solve the issue (ii) for a horizontal 1-strip $R_{j,j'},$ we
consider the rightmost 1-staircase $S_{i,j|i',j'}$ having $R_{j,j'}$ as a basis,
 set $T_{i,j}:=T_{i,j}\cup \{ {\bf t}_{j}\},$ and solve for this staircase the
issue (iii) for  the extended set of terminals.  For all other 1-staircases
$S_{i,j|i',j'}$ and $S_{i',j'|i,j}$ having $R_{j,j'}$ as a basis, we will solve
only the issue (iii) for $T_{i,j}$ and $T_{i',j'},$ respectively.

To deal with (iii), for each vertical 1-strip $R_{i,i'},$ we pick
each shortest path $\pi$ of $\Gamma_k$  between ${\bf t}_i$ and
${\bf t}_{i'},$ include it in the current completion, and solve
(iii) for  all 1-staircases having $R_{i,i'}$ as a vertical base and
taking into account that $\pi$ is already present. We have to
connect the terminals of $T_{i,j}$ by shortest paths of $\Gamma_k$
of least total length to the terminal ${\bf t}_{j'}$ subject to the
condition that the union $\pi\cup s'_{j,j'}$ is already available;
see Fig. \ref{fig_strips_staircases_completion}. For a fixed path
$\pi,$ this task can be done by dynamic programming in
$O(|T_{i,j}|^3)$ time. For this, notice that in an optimal solution
(a) either the highest terminal of $T_{i,j}$ is connected by a
vertical segment to $s'_{j,j'},$ or (b) the lowest terminal of
$T_{i,j}$  is connected by a horizontal segment to $\pi$, or (c)
$T_{i,j}$ contains two consecutive (in the staircase) terminals
${\bf t}_l, {\bf t}_{l+1},$ such that  ${\bf t}_l$ is connected to
$\pi$ by a horizontal segment and ${\bf t}_{l+1}$ is connected to
$s'_{j,j'}$ by a vertical segment.  In each of the three cases and
subsequent recursive calls, we are lead to solve subproblems of the
following type: given a set $T'$ of consecutive terminals of
$T_{i,j},$ the path $\pi$ and a horizontal segment $s',$ connect to
${\bf t}_{j'}$ the terminals of $T'$ by shortest paths of least
total length if the union $\pi\cup s'$ is available. We define by
$C^{\pi}_{i,i'}$ the optimal completion obtained by solving by
dynamic programming those problems for all staircases having
$R_{i,i'}$ as a vertical basis (note that $\pi\subseteq
C^{\pi}_{i,i'}$ however $S^1_h\cap C^{\pi}_{i,i'}=\emptyset$).  For
each vertical 1-strip $R_{i,i'},$ the completion algorithm returns
the partial completion $C^{opt}_{i,i'}$ of least total length, i.e,
$C^{opt}_{i,i'}$  is the smallest completion of the form
$C^{\pi}_{i,i'}$ taken over all $O(n)$ shortest paths $\pi$ running
between ${\bf t}_i$ and ${\bf t}_{i'}$ in $\Gamma_k.$ Finally, let
$N_1^h$ be the union of all $C^{opt}_{i,i'}$ over all vertical
1-strips $R_{i,i'}$ and $S^h_1.$ The pseudocode of the completion
algorithm is presented below (the total complexity of this algorithm
is $O(n^3)$).

\begin{algorithm}
\caption{{\sf Optimal completion$(S_1^h)$}}
\begin{algorithmic}[1]
\STATE $N_1^h \leftarrow S_1^h$ \FOR {each vertical 1-strip
$R_{i,i'}$} \FOR {each shortest path $\pi$ of $\Gamma_k$ connecting
the terminals ${\bf t}_i$ and ${\bf t}_{i'}$} \STATE compute the
partial completion $C^{\pi}_{i,i'}$ in the following way: \STATE
$C^{\pi}_{i,i'} \leftarrow \pi \setminus S_1^h$ \FOR {each
1-staircase $S_{i,j|i',j'}$ and each 1-staircase $S_{i',j'|i,j}$}
\STATE {\bf if} $S_{i,j|i',j'}$ is the rightmost staircase having
$R_{j,j'}$ as a base, {\bf then} set $T_{i,j}\leftarrow T_{i,j}\cup
\{ {\bf t}_j\}$
\STATE compute by dynamic programming the subset $C$
of edges of $\Gamma_k$ of least total length such that $C \cup (\pi
\cup s'_{j,j'})$ contains a shortest path of $\Gamma_k$ from each
terminal of $T_{i,j}$ to ${\bf t}_{j'}$ or from each terminal of
$T_{i',j'}$ to ${\bf t}_{j}$ \STATE $C^{\pi}_{i,i'} \leftarrow
C^{\pi}_{i,i'} \cup C$ \ENDFOR \ENDFOR \STATE let $C^{opt}_{i,i'}$
be the partial completion of least total length, i.e,
$C^{opt}_{i,i'}$  is the smallest completion $C^{\pi}_{i,i'}$ over
all shortest paths $\pi$ between  ${\bf t}_i$ and ${\bf t}_{i'}$
\STATE $N_1^h \leftarrow N_1^h \cup C^{opt}_{i,i'}$ \ENDFOR \STATE
return $N_1^h$
\end{algorithmic}
\end{algorithm}

\begin{lemma} \label{completion} The network $N^h_1$ returned by the algorithm {\sf Optimal completion} is
an optimal completion  for $S_1^h.$
\end{lemma}

\begin{proof} We described above how to compute for each 1-staircase $S_{i,j|i',j'}$
a subset $C$ of edges of $\Gamma_k$ of minimum
total length such that $C \cup (\pi \cup s'_{j,j'})$ contains a
shortest path of $\Gamma_k$ from each terminal of $T_{i,j}$ to ${\bf
t}_{j'}.$ This standard dynamical programming approach explores all
possible solutions and therefore achieves optimality for this
problem. Next, we assert that, for each vertical 1-strip $R_{i,i'},$
the subset of edges $C^{opt}_{i,i'}$ computed by our algorithm, is
an optimal completion of $S_1^h$ for the strip $R_{i,i'}$ and the
staircases having $R_{i,i'}$ as vertical bases. Indeed, our
algorithm considers every possible shortest path $\pi$ of $\Gamma_k$
between ${\bf t}_{i}$ and ${\bf t}_{i'}.$ Once the path $\pi$ is
fixed, the subproblems related to distinct staircases become
independent and can be solved optimally by dynamic programming. The
problems arising from distinct vertical 1-strips are also disjoint
and independent (according to Lemmas \ref{strip_staircase} and
\ref{staircase_staircase}). Therefore the solution $N_1^h$ obtained by
combining the optimal solutions $C^{opt}_{i,i'}$ of every vertical
1-strip $R_{i,i'}$ is an optimal completion of $S_1^h.$

It remains to show that to obtain a  completion satisfying the conditions
(i),(ii), and (iii),  it suffices for each horizontal 1-strip $R_{j,j'}$ to add
${\bf t}_j$ to the set $T_{i,j}$ of terminals of the rightmost staircase
$S_{i,j|i',j'}$ having $R_{j,j'}$ as a basis and to solve (iii) for this extended set of terminals.
 Indeed, in any completion any shortest path between ${\bf t}_j$ and ${\bf t}_{j'}$
 necessarily makes a vertical switch either before  arriving at the origin ${\bf o}$ of $S_{i,j|i',j'}$
 or this path traverses the vertical basis of this staircase.  Since the completion contains a shortest path connecting the terminals of the vertical basis of $S_{i,j|i',j'}$, combining these two paths, we can derive a shortest path between ${\bf t}_j$ and ${\bf t}'_j$ which turns in $R_{i,i'}\cap R_{j,j'}.$ As a result,  we conclude that at least one shortest path between ${\bf t}_j$ and ${\bf t}_{j'}$ passes via ${\bf o}'.$  This shows that indeed it suffices to take into account the condition (ii) only for each rightmost staircase.  $\Box$
\end{proof}

\begin{lemma} \label{admissible}
The network $N_k(T)$ is an admissible solution for the problem  1-DMMN$(F_k).$
\end{lemma}

\begin{proof}
By Lemma \ref{completion}, $N^h_1$ is a completion of $S_1^h$ and
thus contains a shortest path between every pairs of vertices from
$F.$ By symmetry, we get the same result for $N_2^h,$ $N_1^v$ and
$N_2^v.$ Since $N_k(T)$ is one of these networks, by Lemma
\ref{generating_set}, it is admissible solution for the problem
1-DMMN$(F_k).$ $\Box$
\end{proof}

\section{Approximation ratio and complexity}

In this section, we will prove the following main result:

\begin{theorem} \label{2.5BMMN} The network $N_k(T)$ is a factor 1.25
approximation
for 1-DMMN$(F_k)$ problem for $k=0,\ldots,m-1.$ The network
$N(T):=\bigcup_{k=0}^{m-1} N_k(T)$ is a factor 2.5 approximation for
the ${\mathcal B}$-MMN problem and can be constructed in $O(mn^3)$
time.
\end{theorem}

%\begin{proposition} \label{1.25DMMN} The network $N_k(T)$ is a factor 1.25
%approximation for the Minimum
%1-Directional Manhattan Network problem for $F_k,$ $k=0,\ldots,m-1.$
%\end{proposition}

%\begin{theorem} \label{2.5BMMN} The network $N(T):=\bigcup_{k=0}^{m-1} N_k(T)$
%is a factor 2.5 approximation
%for the Minimum ${\mathcal B}$-Manhattan network problem. The
%network $N(T)$ can be constructed in $O(mn^4)$ time.
%\end{theorem}

\noindent{\it Proof of Theorem \ref{2.5BMMN}.} First we prove the
first assertion of the theorem. Let $\Lambda_h=l(S_1^h)=l(S^h_2)$
and $\Lambda_v=l(S_1^v)=l(S_2^v).$ Further, we suppose that
$\Lambda_h\le \Lambda_v.$ Assume $N^{opt}_k$ be an optimal
1-restricted Manhattan network for $F_k.$ Let $M$ be a subnetwork of
$N^{opt}_k\cap (S^h_1\cup S^h_2)$ of minimum total length which
completed with some vertical edges of $N^{opt}_k$ contains a
shortest path between each pair of terminals defining a horizontal
1-strip of $F_k.$ Such $M$ exists because the network $N^{opt}_k\cap
(S^h_1\cup S^h_2)$ already satisfies this requirement. Further, we
assume that $l(M\cap S^h_1)\ge l(M\cap S^h_2).$

\begin{lemma}  $l(M)=\Lambda_h.$
\label{lemma_equal_length}
\end{lemma}

\begin{proof} By Lemma \ref{strips}, two horizontal 1-strips either are disjoint or intersect only in common terminals, thus any horizontal 1-strip $R_{i,i'}$ contributes to $M$ separately from other horizontal 1-strips. Since the terminals ${\bf t}_i$ and ${\bf t}_i'$ defining $R_{i,i'}$ are connected in $N^{opt}_k$ by a shortest path consisting of two horizontal segments of total length equal to the length of a side of $R_{i,i'}$ and a vertical switch between these segments, from the optimality choice of $M$ we conclude that the contribution of $R_{i,i'}$ to $M$ is precisely the length of one of its sides. $\Box$
\end{proof}

\begin{lemma}  $l(M \cap S^h_1) \ge \text{\it 0.5 } l(M).$
\label{lemma_at_least_one_half}
\end{lemma}
\begin{proof}
The proof follows from the assumption $l(M\cap S^h_1)\ge l(M\cap S^h_2)$ and the fact
that $M\cap S^h_1$ and $M\cap S^h_2$ form a partition of $M.$ $\Box$
\end{proof}

\begin{lemma}  $l(S^h_1 \setminus M) \le \text{\it 0.25 } l(N^{opt}_k).$
\label{lemma_at_most_one_fourth}
\end{lemma}
\begin{proof}
Since $l(S^h_1 \setminus M) = l(S^h_1) - l(M\cap S^h_1),$ by Lemma
\ref{lemma_equal_length} and \ref{lemma_at_least_one_half} we get $l(S^h_1
\setminus M) \le
\text{\it 0.5 }l(M) = \text{\it 0.5 } \Lambda_h \le \text{\it 0.25 }
l(N^{opt}_k).$ The last inequality follows from $l(N^{opt}_k)\ge
\Lambda_h+ \Lambda_v$ (a consequence of Lemma \ref{strips}) and the
assumption $\Lambda_h\le \Lambda_v.$ $\Box$
\end{proof}

%\medskip\noindent {\bf Proof of Proposition \ref{1.25DMMN}.}
Now, we complete the proof of Theorem \ref{2.5BMMN}. Note that
$$l(N_k(T))\le l(N^h_1)=l(S_1^h\cup N_1^h) \le l(S_1^h\cup N_k^{opt}) = l(S_1^h\setminus N_k^{opt}) + l(N_k^{opt}) \le \text{\it 1.25 } l(N_k^{opt}).$$
The first inequality follows from the choice of $N_k(T)$ as the
shortest network among the four completions
$N_1^h,N_2^h,N_1^v,N^v_2.$  The second inequality follows from Lemma
\ref{completion} and the fact that $N_k^{opt}$ (and therefore
$S_1^h\cup N_k^{opt}$) is an admissible solution for the completion
problem for $S_1^h.$ Finally, the last inequality follows from Lemma
\ref{lemma_at_most_one_fourth} by noticing that $M\subseteq
N_k^{opt}$  and thus $l(S_1^h\setminus N_k^{opt})\le
l(S_1^h\setminus M)\le \text{\it 0.25 } l(N^{opt}_k).$ This concludes the proof
of the first assertion of Theorem \ref{2.5BMMN}.

%\medskip\noindent {\bf Proof of Theorem \ref{2.5BMMN}.}
 Let  $N^*(T)$ be a
minimum  ${\mathcal B}$-Manhattan network.  For  each direction
$D_k,$ let $N^*_k(T)$ be the set of $k$-segments and
$(k+1)$-segments of $N^*(T).$ By Lemma \ref{shortest_legal_path} the
network $N^*_k(T)$ is an admissible solution for 1-DMMN$(F_k)$
problem, thus $l(N^*_k(T))\ge {\rm OPT}_k(T).$  Any $k$-segment of
$N^*(T)$ belongs to exactly two one-directional networks
$N^*_{k}(T)$ and $N^*_{k-1}(T),$ we conclude that
$\sum_{k=0}^{m-1}{\rm OPT}_k(T)\le \sum_{k=0}^{m-1}l(N^*_k(T))\le
\text{\it 2 } l(N^*(T))= \text{\it 2 }{\rm OPT}(T).$ The first assertion of Theorem \ref{2.5BMMN}
implies that $l(N_k(T))\le
\text{\it 1.25 } l(N_k^{opt})=\text{\it 1.25 }{\rm OPT}_k(T)$ for
all $k=0,\ldots,m-1,$ hence $$l(N(T))\le \sum_{k=0}^{m-1}
l(N_k(T))\le \text{ \it 1.25 }\sum_{k=0}^{m-1}{\rm OPT}_k(T)\le
\text{ \it 2.5 }{\rm OPT}(T).$$ This concludes the proof that the approximation
factor of $N(T):=\bigcup_{k=0}^{m-1} N_k(T)$ is 2.5.

To finish the proof of Theorem \ref{2.5BMMN}, it remains to analyze the complexity of the algorithm. First, we
use a straightforward analysis to establish a $O(mn^4)$
bound on its running time. Then we show that this bound can be reduced
to $O(mn^3)$ by using a more advanced implementation. The time complexity
of {\sf Optimal completion$(S_1^h)$} is dominated by the
execution of the dynamic programming algorithm that computes an
optimal completion for each staircase $S_{i,j|i',j'}$. The staircase
$S_{i,j|i',j'}$ is processed $O(|T_{i,j}|)$ times (once for each
shortest $({\bf t}_i,{\bf t}_{i'})$-path in $\Gamma_k$) using a
$O(|T_{i,j}|^3)$-time dynamic programming algorithm (each of the
$O(|T_{i,j}|^2)$ entries of the dynamic programming table is
computed in time $O(|T_{i,j}|)$). Therefore, each staircase
$S_{i,j|i',j'}$ contributes $O(|T_{i,j}|^4)$ to the execution of the
algorithm {\sf Optimal completion$(S_1^h)$}. Since each terminal
belongs to at most two staircases, the overall complexity of {\sf
Optimal completion$(S_1^h)$} is $O(n^4).$ This algorithm is
processed to compute four optimal completions for each of the $m$
directions. Therefore the total complexity of our 2.5-approximation
algorithm for the ${\mathcal B}$-MMN problem is $O(mn^4).$

The following simple idea allows to reduce the contribution of
each staircase $S_{i,j|i',j'}$ to $O(|T_{i,j}|^3)$ instead of $O(|T_{i,j}|^4),$
leading to a total complexity of $O(mn^3).$ First, note that among all
$O(|T_{i,j}|^2)$ subproblems, whose optima are stored in the dynamic
programming table, only $O(|T_{i,j}|)$ are affected by the choice
of the $\pi$ (those are the subproblems containing the highest and the rightmost
terminal of $T_{i,j}$). Therefore, instead of solving each of $O(|T_{i,j}|^2)$ subproblems
$O(|T_{i,j}|)$ times, we solve the subproblems not affected by the choice of $\pi$ only once. Now,
consider the number of subproblems obtained
by taking into account the choice of $\pi,$ then it is easy to verify
that the total number of subproblems encountered is not $O(|T_{i,j}|^3)$
but only $O(|T_{i,j}|^2).$ Since each entry of the dynamic programming
table is computed in time $O(|T_{i,j}|),$ we obtain a contribution
of $O(|T_{i,j}|^3)$ for each staircase $S_{i,j|i',j'},$
and thus a total complexity of $O(mn^3).$
$\Box$

\section{Conclusion}

In this paper, we presented a combinatorial factor 2.5 approximation
algorithm for  NP-hard minimum Manhattan network problem in normed
planes with polygonal unit balls (the ${\mathcal B}$-MMN problem).
Its complexity is $O(mn^3),$ where $n$ is the number of terminals
and $2m$ is the number of extremal points of the unit ball
${\mathcal B}$. Any ${\mathcal B}$-Manhattan network $N(T)$ can be
decomposed into $m$ subnetworks, one for each direction of the
normed plane. Each such subnetwork $N_k(T)$ ensures the existence of
shortest paths between the pairs of terminals for which all legal
paths use  only $k$- and $(k+1)$-segments. We presented a factor
1.25 $O(n^3)$ algorithm for computing one-directional Manhattan
networks, which lead  to a factor 2.5 algorithm for minimum
${\mathcal B}$-Manhattan network problem. One of the open questions
is whether {\it the one-directional Manhattan network problem is
NP-complete?}  Another open question is  {\it designing a factor 2
approximation algorithm for ${\mathcal B}$-MMN,} thus meeting the
current best approximation factor for the classical MMN problem.
Notice that  polynomial time algorithm for 1-DMMN problem will
directly lead to a factor 2 approximation for ${\mathcal B}$-MMN.

Notice some similarity between the 1-DMMN problem  and the oriented
minASS problem investigated in relationship with the minimum
stabbing box problem \cite{Ha}, alias the minimum arborally
satisfied superset problem (minASS) \cite{DeHaJaKaPa}. In the {\it
minASS problem}, given a set of $n$ terminals $T\subset {\mathbb
R}^2,$ one need to add a minimum number of points $S$ such that for
any pair ${\bf t}_i,{\bf t}_j\in T\cup S,$ either ${\bf t}_i{\bf
t}_j$ is a horizontal or a vertical segment, or the (axis-parallel)
rectangle $R_{i,j}$ spanned by ${\bf t}_i,{\bf t}_j$ contains a
third point of $T\cup S.$ The {\it oriented minASS problem} is
analogous to minASS problem except that the above requirement holds
only for pairs ${\bf t}_i,{\bf t}_j\in T\cup S$ such that $\{
i,j\}\in F_0,$ i.e., ${\bf t}_i$ and ${\bf t}_j$ lie in the first
and the third quadrants of the plane with the same origin. The
authors of \cite{DeHaJaKaPa} presented a polynomial primal-dual
algorithm for oriented minASS problem, however, in contrast to ${\mathcal B}$-MMN problem,
solving oriented
minASS problems for pairs of $F_0$ and $F_1$ (where $F_0\cup
F_1=(T\cup S)\times (T\cup S)$) does not lead to an admissible
solution and thus to a constant factor approximation  for minASS
(which, as we have shown before, is the case for 1-DMMN and
${\mathcal B}$-MMN problems).

\end{document}